\documentclass[journal,draftcls, onecolumn]{IEEEtran}
\usepackage[dvipdfmx]{graphicx}
\usepackage{amsmath}
\usepackage{amssymb}
\usepackage{amsfonts}
\usepackage{mathrsfs}
\usepackage{theorem}
\usepackage{color}
\DeclareMathAlphabet{\bm}{OML}{cmm}{b}{it}
\usepackage[ruled,vlined]{algorithm2e}
\usepackage{algorithmic}

\theorembodyfont{\rmfamily}
\newtheorem{theorem}{Theorem}
\newtheorem{lemma}{Lemma}
\newtheorem{definition}{Definition}
\newtheorem{corollary}{Corollary}
\newtheorem{remark}{Remark}
\newtheorem{proposition}{Proposition}

\newcommand{\qed}{\hfill \IEEEQED}

\newcommand{\bol}[1]{\mathbf{#1}}
\newcommand{\rom}[1]{\mathrm{#1}}

\newcommand{\PI}{\rom{P}_\mathtt{I}}
\newcommand{\PII}{\rom{P}_\mathtt{II}}
\newcommand{\dvar}{d}
\newcommand{\ep}{\varepsilon}

\newcommand{\cC}{{\mathcal C}}
\newcommand{\cD}{{\mathcal D}}

\newcommand{\cP}{{\mathcal P}}

\newcommand{\cS}{{\mathcal S}}
\newcommand{\cT}{{\mathcal T}}

\newcommand{\cX}{{\mathcal X}}
\newcommand{\cY}{{\mathcal Y}}
\newcommand{\cZ}{{\mathcal Z}}

\begin{document}

\title{
Minimax Converse for Identification via Channels
}

\author{Shun Watanabe}


\maketitle
\begin{abstract}
A minimax converse for the identification via channels is derived. By this converse, a general formula
for the identification capacity, which coincides with the transmission capacity, is proved without the assumption of the strong converse property.
Furthermore, the optimal second-order coding rate of the identification via channels is characterized
when the type I error probability is non-vanishing and the type II error probability is vanishing.
Our converse is built upon the so-called partial channel resolvability approach; however, the minimax argument enables us to
circumvent a flaw reported in the literature.
\end{abstract}

\section{Introduction} \label{section:introduction}

The identification is one of typical functions such that randomization significantly reduces the amount 
of communication necessary to compute those functions; eg. see \cite{Kushilevitz-Nisan}.
Inspired by the work by Ja'Ja' \cite{Jaja85}, Ahlswede and Dueck studied
the problem of identification via noisy channels in \cite{AhlDue89, AhlDue89b}; they have shown that,
with randomization, messages of doubly exponential size  of the block-length can be identified, and the optimal 
coefficient is given by Shannon's transmission capacity. Since then, the problem of identification in the context of
information theory has been studied extensively in the literature \cite{HanVer92, han:93, steinberg:98, Bur00, SteMer:02, ahlswede:02, hayashi:06b, oohama:13, YamUed15, BraLap17}; see \cite{ahlswede:08} for a thorough review. 

In many cases, the difficulties of identification problems arise in proving converse coding theorems. 
Initially, the so-called soft converse was proved in \cite{AhlDue89}; the converse coding theorem was only proved 
under the assumption that the identification error probabilities converge to zero exponentially fast in the block-length. 
Later, Han and Verd\'u proved the strong converse coding theorem of the identification via channels in \cite{HanVer92}.
The crucial step of the proof in \cite{HanVer92} is that we replace general stochastic encoders with stochastic encoders having
specific forms, termed ``$M$-types." 
In \cite{han:93}, Han and Verd\'u further studied this step as a separate problem, which they termed the
channel resolvability, by introducing the information spectrum approach.  

The information spectrum approach provides effective tools to investigate coding problems for 
general non-ergodic sources/channels; see \cite{han:book} for a thorough treatment. 
For the channel resolvability, the optimal rate is upper bounded by the spectral sup-mutual information rate
maximized over input processes. On the other hand, the identification capacity of general channels can be
lower bounded by the spectral inf-mutual information rate maximized over input processes.
When those upper bound and lower bound coincide, which is termed the strong converse property, 
it was shown in \cite{han:93} that the identification capacity and the optimal rate of the channel resolvability
coincide with the transmission capacity of the same channels. Later, it was proved in \cite{hayashi:06b} that,
without the assumption of the strong converse property,
the optimal rate of the channel resolvability is characterized by the spectral sup-mutual information rate
maximized over input processes. 

In an attempt to determine the identification capacity without the assumption of the strong converse property,
Steinberg introduced the partial channel resolvability \cite{steinberg:98}. 
In the partial channel resolvability, we consider a truncated channel so that 
the tail probability of information spectrum is not accumulated twice 
in the argument of relating the channel resolvability to the identification code.
It should be noted that, in the modern terminology, considering the partial response is
essentially equivalent to the technique termed ``smoothing" \cite{renner:05b}. For instance,
the channel resolvability for smoothed channels has been effectively used to derive second-order
bounds on coding problems with side-information \cite{watanabe:13e}.

Using the partial channel resolvability, it was claimed in \cite{steinberg:98} that 
the identification capacity of general channels coincides with the transmission capacity of the same channels. 
However, there is a flaw in the proof of \cite[Lemma 2]{steinberg:98},
which has been reported in \cite[Remark 2]{hayashi:06b}.
Thus, without the assumption of the strong converse property, the identification capacity of general channels has been
an open problem so far. The main purpose of this paper is to provide a remedy to the result claimed in \cite{steinberg:98}.
In fact, our converse is built upon the partial channel resolvability; however, in order to circumvent the aforementioned flaw,
we leverage the minimax argument described below. 

In the past few decades, the argument based on the hypothesis testing has been successfully used to derive
a converse bound on transmission codes of general channels \cite{hayashi:03, polyanskiy:10}, which is termed the meta converse.\footnote{For a detailed
historical perspective on the meta converse, see \cite{Hayashi17}.}
Particularly, a useful feature of the meta converse bound is that we can choose an auxiliary 
output distribution; thus, the expression of the converse bound involves the minimum over the output distribution and the maximum over the input distribution.
For the asymptotic analysis of discrete memoryless channels, the Shannon capacity is recovered from the minimax expression 
by the Topsoe identity \cite{Topsoe:67}.
In fact, the flexibility of choosing the output distribution has been effectively used to
derive finer asymptotic results: the second-order coding rate \cite{hayashi:09, polyanskiy:10} and the third-order coding rate \cite{TomTan:13};
see also \cite{Tan:book}. Also, Polyanskiy proved that the order of minimax
in the meta converse bound can be interchanged under certain regularity conditions \cite{polyanskiy:13}.

In this paper, we derive a minimax converse bound for the identification via channels. 
To that end, we utilize a modified version \cite {hayashi-book:17} of the so-called soft covering lemma reported in \cite{hayashi:06b, oohama:13, cuff:12};
the modified bound on the channel resolvability involves an auxiliary output distribution.
The main contribution of this paper is to apply the flexibility of choosing the auxiliary output distribution 
to the argument connecting the channel resolvability and the identification code.\footnote{Recently, the flexibility of choosing the auxiliary output distribution
was used in a different manner to derive the identification capacity of the covert communication \cite{ZhaTan20}.}
The key difference between our argument and the argument in \cite[Lemma 2]{steinberg:98}
is as follows: in our argument, we consider a truncated channel induced from a fixed auxiliary output distribution;
on the other hand, truncated channels are constructed from output distributions that depend on
input distributions in \cite[Lemma 2]{steinberg:98}. In the former case, we can bound the number of messages 
of an identification code by the number of $M$-types without causing any trouble; this enables us to circumvent the flaw reported in \cite[Remark 2]{hayashi:06b}. 
See Remark \ref{remark:proof-argument} of Section \ref{section:identification-channel-resolvability} for more detail.  

By using the minimax converse bound, we derive the identification capacity of general channels;
it turns out that the identification capacity coincides with the transmission capacity without the assumption of the strong converse property. 
In the derivation of this result, we invoke the aforementioned result in \cite{polyanskiy:13} to interchange the order of
the minimum over the output distribution and the maximum over the input distribution.
Furthermore, we also derive the optimal second-order coding rate of the identification via channels 
when the type I error probability is non-vanishing and the type II error probability is vanishing.

\paragraph*{Notation}

Throughout the paper, random variables (eg.~$X$) and their realizations (eg.~$x$)
are denoted by capital and lower case letters, respectively. All random variables take values in some
finite alphabets which are denoted by the respective calligraphic letters (eg.~$\cX$). The probability 
distribution of random variable $X$ is denoted by $P_X$. Similarly, $X^n=(X_1,\ldots,X_n)$
and $x^n = (x_1,\ldots,x_n)$ denote, respectively, a random vector and its realization in the $n$th
Cartesian product $\cX^n$.
For a finite set $\cS$, the cardinality of $\cS$ is denoted by $|\cS|$. 
For a subset $\cT \subseteq \cS$, the complement $\cS \backslash \cT$ is denoted by $\cT^c$. 
The set of all distributions on $\cX$ is denoted by $\cP(\cX)$. The indicator function is denoted by $\bol{1}[\cdot]$.
Information theoretic quantities are denoted in the same manner as \cite{cover, csiszar-korner:11, han:book}.
All information quantities and rates are evaluated with respect to the natural logarithm.
For given sub-distributions $P$ and $Q$ that are not necessarily normalized, the variational distance is denoted by $d(P,Q):= \frac{1}{2}\sum_x |P(x)-Q(x)|$.

\section{Problem Formulation of Identification via Channels} \label{section:problem-formulation}

In this section, we describe the problem formulation of the identification via channels, and review basic results.
We start with the problem formulation for the single-shot regime. 
Given a channel $W$ from ${\cal X}$ to ${\cal Y}$, the sender tries to transmit one of
$N$ messages; then the receiver shall identify if message $i \in \{1,\ldots,N\}$ was transmitted or not.
The encoder is given by stochastic mappings $P_1,\ldots,P_N \in {\cal P}({\cal X})$, and the decoder is given by
acceptance regions ${\cal D}_1,\ldots,{\cal D}_N \subset {\cal Y}$ for each message. Note that, unlike the standard 
transmission code, the acceptance regions of an identification code need not be disjoint. 
In other words, if the receiver is intended to identify message $i$, there is no need to distinguish 
messages other than $i$. 

For a given identification code $\{ (P_i, {\cal D}_i)\}_{i=1}^N$, the first type error probability is given by
\begin{align*}
\PI := \max_{1 \le i \le N} P_iW({\cal D}_i^c ),
\end{align*}
and the second type error probability is given by
\begin{align*}
\PII := \max_{1 \le i \neq j \le N} P_iW({\cal D}_j ),
\end{align*}
where 
\begin{align*}
P_iW(y) = \sum_{x \in {\cal X}} P_i(x) W(y|x)
\end{align*}
is the output distribution of the channel $W$ corresponding to the input distribution $P_i$.
For given error probabilities $0 \le \varepsilon,\delta < 1$ with $\varepsilon + \delta < 1$, 
an identification code $\{ (P_i, {\cal D}_i)\}_{i=1}^N$ is called $(N,\ep, \delta)$-ID code for channel $W$ if $\PI \le \ep$ and $\PII \le \delta$
are satisfied. 
Then, the optimal code size of identification via channel $W$ is defined by
\begin{align*}
N^\star(\ep,\delta|W) := \sup\big\{ N : \mbox{$(N,\ep,\delta)$-ID code for $W$ exists}  \big\}.
\end{align*}

When we consider the block coding over $n$ uses $W^n$ of a channel, 
it is known that the optimal code size $N^\star(\ep,\delta|W^n)$ grows doubly exponentially 
in the block length $n$. For a discrete memoryless channel, it has been known that
the identification capacity
\begin{align*}
C_{\mathtt{ID}}(\ep,\delta|W) := \liminf_{n\to\infty} \frac{1}{n} \log\log N^\star(\ep,\delta|W^n)
\end{align*}
coincide with the transmission capacity \cite{AhlDue89, HanVer92}, i.e., 
\begin{align*}
C_{\mathtt{ID}}(\ep,\delta|W) = C(W) := \max_{P_X} I(X\wedge Y)
\end{align*}
as long as $\ep+\delta<1$, where $I(X \wedge Y)$ is the mutual information between $(X,Y)$ distributed 
according to $P_{XY}(x,y)=P_X(x)W(y|x)$.
It should be noted that the identification capacity is infinite 
when $\ep+\delta \ge 1$ \cite{HanVer92}.

\section{Hypothesis Testing} \label{section:hypothesis-testing}

In this section, we summarize known facts on the hypothesis testing and the meta converse
that are needed in the rest of the paper. 
Consider a binary hypothesis testing with a null hypothesis $Z \sim P_Z$
and an alternative hypothesis $Z \sim Q_Z$, where $P_Z$ and $Q_Z$ are distribution 
on the same alphabet $\cZ$. Upon observing $Z=z$, we shall decide whether the 
value was generated by the distribution $P_Z$ or the distribution $Q_Z$.
Most general test can be described by a channel $T$ from $\cZ$ to $\{0,1\}$,
where $0$ indicates the null hypothesis and $1$ indicates the alternative hypothesis. 
When $z \in \cZ$ is observed, the test $T$ chooses the null hypothesis with probability $T(0|z)$
and the alternative hypothesis with probability $T(1|z)=1-T(0|z)$. 
Then, the type I error probability of the test is defined by
\begin{align*}
\PI[T] := \sum_z P_Z(z) T(1|z),
\end{align*}
and the type II error probability of the test is defined by
\begin{align*}
\PII[T] := \sum_z Q_Z(z) T(0|z).
\end{align*}
For a given $0 \le \ep < 1$, denote by $\beta_\ep(P_Z,Q_Z)$ the 
optimal type II error probability under the condition that the type I error probability is less than $\ep$, i.e.,
\begin{align*}
\beta_\ep(P_Z,Q_Z) := \inf_{T : \atop \PI[T]\le \ep} \PII[T].
\end{align*}
In fact, since $\beta_\ep(P_Z,Q_Z)$ can be described as a linear programming when $\cZ$ is finite,
the infimum can be attained. 

For a threshold parameter $\gamma \in \mathbb{R}$, the test given by
\begin{align*}
T(0|z) = \bol{1}\bigg[ \log \frac{P_Z(z)}{Q_Z(z)} > \gamma\bigg]
\end{align*}
is termed the likelihood ratio test, also known as the Neyman-Pearson test.
For given $0 \le \ep <1$, let  
\begin{align*}
D_{\mathtt{s}}^\ep(P_Z \|Q_Z) := \sup\bigg\{ \gamma \in \mathbb{R} : \Pr\bigg( \log \frac{P_Z(Z)}{Q_Z(Z)} \le \gamma \bigg) \le \ep \bigg\},
\end{align*}
where the probability is with respect to $Z \sim P_Z$. Note that
the quantity is the supremum of thresholds such that the type I error probability of the likelihood ratio test is less than $\ep$,
and it is referred to as $\ep$-information spectrum divergence \cite{tomamichel:12}. 
This quantity and the optimal type II error probability defined above have the following relationship (eg.~see \cite[Lemma 2.4]{Tan:book});
it can be understood as a variant of the Neyman-Pearson lemma claiming that the likelihood ratio test is essentially optimal. 
\begin{lemma} \label{lemma:connection-Ds-beta}
For a given $0 \le \ep < 1$, it holds that
\begin{align*}
D_{\mathtt{s}}^\ep(P_Z \|Q_Z) 
\le -\log\beta_\ep(P_Z,Q_Z) \le D_{\mathtt{s}}^{\ep+\zeta}(P_Z\|Q_Z) + \log(1/\zeta)
\end{align*}
for any $0< \zeta < 1-\ep$.
\end{lemma}
This lemma enables us to use the two quantities almost interchangeably. 

As we have mentioned in Section \ref{section:introduction}, in the past few decades, 
the hypothesis testing has become a useful tool to derive a converse bound on transmission codes
over a channel $W$ from $\cX$ to $\cY$. 
For such an application, we consider the hypothesis testing between the null hypothesis
\begin{align*}
P\times W(x,y) := P(x) W(y|x)
\end{align*} 
and the alternative hypothesis 
\begin{align*}
P\times Q(x,y) := P(x) Q(y),
\end{align*}
where $P \in \cP(\cX)$ and $Q \in \cP(\cY)$ are given input/output distributions. 
More specifically, the optimal coding rate of transmission codes is bounded in terms of 
\begin{align} \label{eq:optimization-problem}
\inf_{P \in \cP(\cX)} \sup_{Q \in \cP(\cY)} \beta_\ep(P \times W, P\times Q).
\end{align} 
It can be easily verified from the definition that $\beta_\ep(P \times W, P\times Q)$ is concave 
with respect to the output distribution $Q \in \cP(\cY)$. On the other hand, it was proved in \cite{polyanskiy:13} that $\beta_\ep(P \times W, P\times Q)$
is convex with respect to the input distribution $P \in \cP(\cX)$. Thus, $\beta_\ep(P \times W, P\times Q)$ is a convex-concave function
on $\cP(\cX) \times \cP(\cY)$, and regularity conditions on the saddle-point property were discussed in \cite{polyanskiy:13};
particularly, since $\cP(\cX)$ and $\cP(\cY)$ are compact for finite alphabets $\cX$ and $\cY$, the following saddle-point property 
follows from the classic min-max theorem. 
\begin{lemma}[\cite{polyanskiy:13}] \label{lemma:min-max-of-beta}
For a given $0 \le \ep <1$, the optimal value in \eqref{eq:optimization-problem} is attainable and
\begin{align*} 
\min_{P \in {\cal P}({\cal X})} \max_{Q \in {\cal P}({\cal Y})} \beta_\varepsilon(P\times W, P\times Q) 
= \max_{Q \in {\cal P}({\cal Y})} \min_{P \in {\cal P}({\cal X})} \beta_\varepsilon(P\times W, P\times Q).
\end{align*}
\end{lemma}

When we evaluate asymptotic behavior of coding rates for a DMC, it is more convenient to
use the $\ep$-information spectrum divergence. Particularly, we will use 
the following symbol-wise relaxation (eg.~see \cite{TomTan:13}):
\begin{align} \label{eq:symbol-wise-relaxation}
D_{\mathtt{s}}^\ep(P\times W \| P \times Q) \le \max_{x\in {\cal X}} D_{\mathtt{s}}^\ep( W(\cdot |x) \| Q)
\end{align}
for any $Q \in \cP(\cY)$.

\section{Main Result: Minimax Converse for Identification via Channels} \label{section:identification-channel-resolvability}

In this section, we present our main result, i.e., the minimax converse bound on the identification via channels.
To that end, we first explain the problem of channel resolvability. 

For an integer $M$, a distribution $P \in \cP(\cX)$ is said to be an $M$-type
if $P(x)$ is an integer multiple of $1/M$ for every $x\in \cX$.
Then, $(N,\ep,\delta)$-ID code $\{ (P_i,\cD_i)\}_{i=1}^N$ is said to be $M$-canonical if
$P_i$ is an $M$-type for every $1\le i \le N$. For $M$-canonical $(N,\ep,\delta)$-ID code 
with $\ep+\delta<1$, it is not difficult to see that all $P_i$s are distinct; in fact, if there exist $i$ and $j$
such that $P_i = P_j$, then
\begin{align*}
1-\ep \le P_iW(\cD_i) = P_jW(\cD_i) \le \delta,
\end{align*}
which contradict $\ep + \delta <1$. Since the number of $M$-types on $\cX$ 
is at most $|\cX|^M$, we must have $N\le |\cX|^M$ for $M$-canonical ID code. 

In \cite{han:93}, among other motivations, the channel resolvability was introduced as a tool handle
general ID codes by relating their analysis to that of $M$-canonical codes. 
In the channel resolvability problem, we shall approximate the output distribution $PW$
of an arbitrarily given input distribution $P \in \cP(\cX)$ by the output distribution $\tilde{P}W$ of an $M$-type $\tilde{P}$ 
so that 
\begin{align*}
d(\tilde{P}W, PW) \le \zeta
\end{align*}
is satisfied for a prescribed approximation error $\zeta$. If such an approximation is realized, then we can replace 
each $P_i$ with an $M$-type $\tilde{P}_i$, and use the above mentioned counting argument for $M$-canonical codes. 

In an attempt to derive a tighter converse bound than that in \cite{han:93}, 
the partial channel resolvability was introduced in \cite{steinberg:98}. 
For a given subset ${\cal S} \subset {\cal X} \times {\cal Y}$, let us introduce the truncated channel
\begin{align*}
W^\cS(y|x) := W(y|x) \bol{1}[(x,y) \in {\cal S}]
\end{align*}
and the truncated output distribution
\begin{align*}
P W^{\cal S}(y) := \sum_{x \in {\cal X}} P(x) W(y|x) \bol{1}[(x,y) \in {\cal S}].
\end{align*}
Note that $PW^\cS$ is a sub-distribution, i.e., it may not add up to $1$, and 
it is referred to as the partial response of the input distribution $P$.
It can be immediately verified that
\begin{align} \label{eq:truncation-error}
\dvar(PW^{\cal S}, PW) = \frac{P\times W({\cal S}^c)}{2}.
\end{align}
In the partial channel resolvability problem, we shall approximate the partial response $PW^\cS$
of an arbitrarily given input distribution $P \in \cP(\cX)$ by the partial response $\tilde{P}W^\cS$ of 
an $M$-type $\tilde{P}$ so that
\begin{align*}
d(\tilde{P}W^\cS, PW^\cS) \le \zeta
\end{align*}
is satisfied for a prescribed approximation error $\zeta$.

A standard approach of constructing the (partial) channel resolvability code 
is to randomly generate $M$ symbols $x_1,\ldots, x_M$ according to distribution $P$.
The performance analysis of such a random code construction is referred to as the soft covering lemma \cite{cuff:12}.
The following lemma is a variant of the soft covering lemma, and it can be derived in almost the same manner
as \cite{hayashi:06b, oohama:13, cuff:12} with a simple modification. 
Even though the modified version is available in the literature \cite{hayashi-book:17, ZhaTan20}, we provide a proof here for completeness.
\begin{lemma} \label{lemma:soft-covering}
For arbitrarily given $Q \in {\cal P}({\cal Y})$ and $\gamma \in \mathbb{R}$, let
\begin{align} \label{eq:partial-set}
{\cal S} = {\cal S}(Q,\gamma) := \bigg\{ (x,y) \in {\cal X}\times {\cal Y}: \log \frac{W(y|x)}{Q(y)} \le \gamma \bigg\}.
\end{align}
Then, for a given input distribution $P \in {\cal P}({\cal X})$, there exists an $M$-type $\tilde{P}$ such that
\begin{align} \label{eq:soft-covering-statement}
\dvar(\tilde{P}W^{\cal S}, PW^{\cal S}) \le \frac{1}{2} \sqrt{\frac{e^\gamma}{M}}
\end{align}
\end{lemma}
\begin{proof}
Let $\cC = \{X_1,\ldots,X_M\}$ be a codebook such that each $X_i$ is randomly generated with distribution $P$.
Then, we define $M$-type $\tilde{P} = \tilde{P}_\cC$ by 
\begin{align*}
\tilde{P}(x) = \frac{1}{M} \sum_{i=1}^M \bol{1}[X_i = x].
\end{align*} 
We shall evaluate the approximation error averaged over the random generation of the codebook $\cC$.
By Jensen's inequality and the convexity of $t\to t^2$, we have
\begin{align}
\mathbb{E}_{\cC}\big[ d(\tilde{P}W^\cS, PW^\cS) \big]^2 \le \mathbb{E}_\cC\big[ d(\tilde{P}W^\cS, PW^\cS)^2\big]. 
\label{eq:soft-covering-proof-1}
\end{align}
Then, we have
\begin{align}
\lefteqn{ 4 \mathbb{E}\big[ d(\tilde{P}W^\cS, PW^\cS)^2 \big] } \nonumber \\
&= \mathbb{E}_{\cC}\bigg[ \bigg( \sum_y \big| \tilde{P}W^\cS(y) - PW^\cS(y) \big| \bigg)^2 \bigg] \nonumber \\
&= \mathbb{E}_{\cC}\bigg[ \bigg( \sum_y \sqrt{Q(y)}\sqrt{Q(y)} \bigg| \frac{\tilde{P}W^\cS(y) - PW^\cS(y)}{Q(y)} \bigg|  \bigg)^2 \bigg] \nonumber \\
&\le \mathbb{E}_{\cC}\bigg[ \sum_y Q(y) \bigg| \frac{\tilde{P}W^\cS(y) - PW^\cS(y)}{Q(y)} \bigg|^2 \bigg],
 \label{eq:soft-covering-proof-2}
\end{align}
where the summation $y$ is taken over $\mathtt{supp}(Q)$,\footnote{Note that $\tilde{P}W^\cS(y)=PW^\cS(y)=0$ whenever $Q(y)=0$
from the definition of $\cS$.}
and the last inequality follows from the Cauchy-Schwarz inequality. Denoting $Y \sim Q$, we can rewrite the above formula as
\begin{align}
\lefteqn{ \mathbb{E}_{\cC}\bigg[ \sum_y Q(y) \bigg| \frac{\tilde{P}W^\cS(y) - PW^\cS(y)}{Q(y)} \bigg|^2 \bigg] } \nonumber \\
&= \mathbb{E}_Y \mathbb{E}_{\cC}\bigg[ \bigg( \frac{\tilde{P}W^\cS(Y)}{Q(Y)} - \frac{PW^\cS(Y)}{Q(Y)} \bigg)^2 \bigg] \nonumber \\
&= \mathbb{E}_Y \mathbb{E}_{\cC}\bigg[ \bigg( \sum_{i=1}^M \frac{1}{M} \frac{W^\cS(Y|X_i)}{Q(Y)} - \frac{PW^\cS(Y)}{Q(Y)} \bigg)^2 \bigg] \nonumber \\
&=  \mathbb{E}_Y \mathbb{E}_{\cC}\bigg[ \frac{1}{M^2} \sum_{i=1}^M \bigg( \frac{W^\cS(Y|X_i)}{Q(Y)} \bigg)^2
 + \sum_{i,j =1: \atop i\neq j}^M \frac{1}{M^2} \frac{W^\cS(Y|X_i)}{Q(Y)} \frac{W^\cS(Y|X_j)}{Q(Y)}  \nonumber \\
&~~~  - \sum_{i=1}^M \frac{2}{M} \frac{W^\cS(Y|X_i)}{Q(Y)} \frac{PW^\cS(Y)}{Q(Y)} + \bigg( \frac{PW^\cS(Y)}{Q(Y)} \bigg)^2 \bigg]. \label{eq:soft-covering-proof-3} 
\end{align}
Furthermore, by noting that, for $i\neq j$,
\begin{align*}
\mathbb{E}_{\cC}\bigg[ \frac{W^\cS(Y|X_i)}{Q(Y)}  \bigg] &= \frac{PW^\cS(Y)}{Q(Y)}, \\
\mathbb{E}_{\cC}\bigg[ \frac{W^\cS(Y|X_i)}{Q(Y)} \frac{W^\cS(Y|X_j)}{Q(Y)} \bigg] 
 &= \mathbb{E}_{X_i} \bigg[ \frac{W^\cS(Y|X_i)}{Q(Y)} \bigg] \mathbb{E}_{X_j} \bigg[ \frac{W^\cS(Y|X_j)}{Q(Y)} \bigg] = \bigg( \frac{PW^\cS(Y)}{Q(Y)} \bigg)^2,
\end{align*}
we can rewrite \eqref{eq:soft-covering-proof-3} as 
\begin{align}
\lefteqn{ \frac{1}{M} \mathbb{E}_Y \mathbb{E}_{X}\bigg[ \bigg( \frac{W^\cS(Y|X)}{Q(Y)} \bigg)^2 -  \bigg( \frac{PW^\cS(Y)}{Q(Y)} \bigg)^2 \bigg] } \nonumber \\
&\le \frac{1}{M} \mathbb{E}_Y \mathbb{E}_{X}\bigg[ \bigg( \frac{W^\cS(Y|X)}{Q(Y)} \bigg)^2 \bigg] \nonumber \\
&= \frac{1}{M} \sum_{x,y} P(x) \frac{W(y|x)^2}{Q(y)} \bol{1}[(x,y) \in \cS] \nonumber \\
&\le \frac{1}{M} \sum_{x,y} P(x) W(y|x) e^\gamma \bol{1}[(x,y) \in \cS] \nonumber \\
&\le \frac{e^\gamma}{M}, \label{eq:soft-covering-proof-4}
\end{align}
where $X \sim P$.
By combining \eqref{eq:soft-covering-proof-1}, \eqref{eq:soft-covering-proof-2}, \eqref{eq:soft-covering-proof-3}, and \eqref{eq:soft-covering-proof-4},
we have
\begin{align*}
\mathbb{E}_{\cC}\big[ d(\tilde{P}W^\cS, PW^\cS) \big] \le \frac{1}{2}\sqrt{\frac{e^\gamma}{M}},
\end{align*}
which implies the existence of an $M$-type $\tilde{P}$ satisfying \eqref{eq:soft-covering-statement}.
\end{proof}

The difference between Lemma \ref{lemma:soft-covering} and the standard soft covering lemmas
is that we can arbitrarily choose an auxiliary output distribution $Q \in \cP(\cY)$ instead of the output 
distribution $PW$ that corresponds to the input distribution $P$. A similar usage of the auxiliary distribution
has been known in the context of a related problem, the privacy amplification \cite{renner:05b}.

The main innovation of this paper is that we use the above mentioned flexibility of choosing
the auxiliary output distribution to derive a novel converse bound on the ID code.
\begin{theorem} \label{theorem:connection}
For arbitrarily given $Q \in {\cal P}({\cal Y})$ and $\gamma \in \mathbb{R}$, let ${\cal S} = {\cal S}(Q,\gamma)$ be
defined as in \eqref{eq:partial-set}. Then, for an arbitrary integer $M$, any $(N,\ep,\delta)$-ID code with $N > |{\cal X}|^M$ must satisfy
\begin{align} \label{eq:connection-assumption}
\ep + \delta \ge \inf_{P \in {\cal P}({\cal X})} P\times W({\cal S}) - \sqrt{\frac{e^\gamma}{M}}.
\end{align}
\end{theorem}
\begin{proof}
For an arbitrarily given $(N,\ep,\delta)$-ID code $\{ (P_i, {\cal D}_i) \}_{i=1}^N$, we have
\begin{align}
\dvar( P_iW, P_jW) &\ge P_i W({\cal D}_i) - P_j W({\cal D}_i) \nonumber \\
&\ge 1 - \ep - \delta \label{eq:proof-connection-1}
\end{align}
for every $i \neq j$. By applying Lemma \ref{lemma:soft-covering} for each $P_i$,
we can find $M$-type $\tilde{P}_i$ such that 
\begin{align}
\dvar( \tilde{P}_iW^{\cal S}, P_iW^{\cal S}) \le \frac{1}{2} \sqrt{\frac{e^\gamma}{M}}. \label{eq:proof-connection-2}
\end{align}
Since the number of distinct $M$-types is upper bonded by $|{\cal X}|^M$ and since $N > |{\cal X}|^M$ by assumption,
there must exist a pair $i$ and $j$ such that $\tilde{P}_i = \tilde{P}_j$. For such a pair, by applying the triangular inequality twice, we have
\begin{align}
\dvar( P_iW, P_jW) &\le \dvar( P_i W, \tilde{P}_iW^{\cal S}) + \dvar(\tilde{P}_i W^{\cal S}, \tilde{P}_j W^{\cal S}) + \dvar(\tilde{P}_jW^{\cal S},P_j W) \nonumber \\
&= \dvar( P_i W, \tilde{P}_iW^{\cal S}) + \dvar(\tilde{P}_jW^{\cal S},P_j W) \nonumber \\
&\le \dvar( P_i W, P_i W^{\cal S}) + \dvar( P_i W^{\cal S}, \tilde{P}_iW^{\cal S}) + \dvar(\tilde{P}_jW^{\cal S},P_jW^{\cal S}) + \dvar(P_jW^{\cal S},P_jW) \nonumber \\
&\le \frac{P_i\times W({\cal S}^c) + P_j \times W({\cal S}^c)}{2} + \sqrt{\frac{e^\gamma}{M}} \nonumber \\
&\le \sup_{P \in {\cal P}({\cal X})} P\times W({\cal S}^c) + \sqrt{\frac{e^\gamma}{M}}, \label{eq:proof-connection-3}
\end{align}
where the second last inequality follows from \eqref{eq:proof-connection-2} and \eqref{eq:truncation-error}.
Then, \eqref{eq:proof-connection-3} together with \eqref{eq:proof-connection-1} imply \eqref{eq:connection-assumption}.
\end{proof}

\begin{remark}
Without using the partial channel resolvability, it can be proved that 
any $(N,\ep,\delta)$-ID code with $N > |\cX|^M$ must satisfy\footnote{For instance, see Eq.~(17) and Lemma 3 of \cite{hayashi:06b}.}
\begin{align} \label{eq:existing-bound}
\ep + \delta \ge \inf_{P \in \cP(\cX)} \big[ 1- 2 P\times W(\cT_P^c) \big] - \sqrt{\frac{e^\gamma}{M}},
\end{align}
where 
\begin{align} \label{eq:typical-set-PW}
\cT_P = \cT(P,\gamma) := \bigg\{ (x,y) \in \cX \times \cY : \log \frac{W(y|x)}{PW(y)} \le \gamma \bigg\}.
\end{align}
The factor $2$ of the first term in \eqref{eq:existing-bound} has prevented us from deriving a general formula of the ID-capacity
without the strong converse property. 
\end{remark}

\begin{remark} \label{remark:proof-argument}
The proof of Theorem \ref{theorem:connection} is inspired in part from 
the argument in \cite[Lemma 2]{steinberg:98}, 
which has a flaw reported in \cite[Remark 2]{hayashi:06b}. 
A crucial difference between our argument and that in \cite[Lemma 2]{steinberg:98}
is that the set ${\cal S}$ for a fixed $Q$ is used to construct the truncated channel $W^\cS$ in our argument, while 
the set ${\cal T}_{P_i}$ defined by \eqref{eq:typical-set-PW} is used to construct 
the truncated channel $W^{\cT_{P_i}}$ for each $i$ in \cite[Lemma 2]{steinberg:98}.
In the former case, $\tilde{P}_i=\tilde{P}_j$ implies $\tilde{P}_i W^{\cal S} =\tilde{P}_j W^{\cal S}$,
and the size $N$ of the ID code is bounded by the number $|{\cal X}|^M$ of $M$-types eventually.\footnote{More precisely,
we have used the contraposition of this claim.}
On the other hand, in the latter case, we cannot conclude that $\tilde{P}_1,\ldots, \tilde{P}_N$ are all distinct 
since $\tilde{P}_i=\tilde{P}_j$ does not necessarily imply $\tilde{P}_i W^{{\cal T}_{P_i}} =\tilde{P}_j W^{{\cal T}_{P_j}}$;
thus, the size $N$ of the ID code cannot be bounded by the number of $M$-types. Instead, it was attempted in \cite[Lemma 2]{steinberg:98}
to bound $N$ by the number of some alternative measures induced by $M$-types, which has a flaw \cite[Remark 2]{hayashi:06b}. 
\end{remark}

\begin{corollary} \label{corollary:single-shot-converse-Ds-expression}
For $0 \le \ep,\delta < 1$ with $\ep+\delta < 1$ and arbitrary $0 < \eta < 1 - \ep - \delta$, we have
\begin{align} 
\log\log N^\star(\ep,\delta|W) \le \inf_{Q \in {\cal P}({\cal Y})} \sup_{P \in {\cal P}({\cal X})}
D_{\mathtt{s}}^{\ep+\delta+\eta}(P\times W \| P\times Q)
+ \log \log |{\cal X}| + 2 \log(1/\eta) + 2. 
\label{eq:Ds-converse}
\end{align}
\end{corollary}
\begin{proof}
For arbitrary $(N,\ep,\delta)$-ID code,
by setting\footnote{Since $D_{\mathtt{s}}^{\ep+\delta+\eta}(P\times W \| P\times Q) + \log(1/\eta) \ge -\log \beta_{\ep+\delta}(P\times W, P\times Q)\ge 0$
by Lemma \ref{lemma:connection-Ds-beta}, \eqref{eq:Ds-converse} trivially holds if $N \le |{\cal X}|$.
Thus, we only consider the case with $N > |{\cal X}|$, which implies $M \ge 1$.}  
\begin{align*}
M = \bigg\lfloor \frac{\log (N-1)}{\log|{\cal X}|} \bigg\rfloor
\end{align*}
so that $N>|\cX|^M$
and 
\begin{align} \label{eq:proof-Ds-converse-1}
\gamma = 2 \log \eta + \log\log N - \log\log|{\cal X}| -2
\end{align}
so that $\sqrt{e^\gamma/M} \le \eta$, Theorem \ref{theorem:connection} implies
\begin{align} \label{eq:proof-Ds-converse-2}
\inf_{P\in {\cal P}({\cal X})} P\times W({\cal S}) \le \ep + \delta + \eta,
\end{align}
where $\cS = {\cal S}(Q,\gamma)$ is defined as in \eqref{eq:partial-set}
for arbitrarily fixed $Q \in \cP(\cY)$. In fact, since the lefthand side of \eqref{eq:proof-Ds-converse-2}
is linear with respect to $P$ and $\cP(\cX)$ is a compact set, the infimum in \eqref{eq:proof-Ds-converse-2} can be attained for some $P\in \cP(\cX)$. 
This means that $D_{\mathtt{s}}^{\ep+\delta+\eta}(P\times W\|P\times Q) \ge \gamma$ for some $P \in \cP(\cX)$, i.e.,
\begin{align*}
\log \log N \le \sup_{P\in\cP(\cX)} D_{\mathtt{s}}^{\ep+\delta+\eta}(P\times W\|P\times Q) + \log\log|\cX| + 2\log(1/\eta)+2.
\end{align*}
Since this bound holds for arbitrary $(N,\ep,\delta)$-ID code and $Q \in \cP(\cY)$, we have the claim of the corollary.
\end{proof}

From Lemma \ref{lemma:connection-Ds-beta} and Lemma \ref{lemma:min-max-of-beta}, Corollary \ref{corollary:single-shot-converse-Ds-expression} implies the following corollary.
\begin{corollary} \label{corollary:single-shot-converse-beta-expression}
For $0 \le \ep,\delta < 1$ with $\ep+\delta < 1$ and arbitrary $0 < \eta < 1 - \ep - \delta$, we have
\begin{align}
\log\log N^\star(\ep,\delta|W) &\le \min_{Q \in {\cal P}({\cal Y})} \max_{P \in {\cal P}({\cal X})} -\log \beta_{\ep+\delta+\eta}(P\times W, P\times Q) + \log \log |{\cal X}| + 2 \log(1/\eta) + 2  
\label{eq:min-max-converse}\\
&= \max_{P \in {\cal P}({\cal X})} \min_{Q \in {\cal P}({\cal Y})}  -\log \beta_{\ep+\delta+\eta}(P\times W, P\times Q) + \log \log |{\cal X}| + 2 \log(1/\eta) + 2. \label{eq:max-min-converses}
\end{align}
\end{corollary}

Up to some residual terms,
the upper bounds on the doubly exponential rate of the optimal ID code in Corollary \ref{corollary:single-shot-converse-Ds-expression} 
and Corollary \ref{corollary:single-shot-converse-beta-expression} have the same form as the upper bounds
on the rate of the optimal transmission code reported in the literature \cite{polyanskiy:10}. 
In the next section, we will discuss asymptotic behaviors of those bounds.

\section{Capacity for General Channels}

In this section, we derive the identification capacity of general channels.
Let $\bm{W} = \{ W^n \}_{n=1}^\infty$ be a sequence of general channels from 
$\cX^n$ to $\cY^n$, where $\cX$ and $\cY$ are finite alphabets; the channel $\bm{W}$ may not be stationary nor ergodic. 
For each integer $n$, an $(N_n,\ep_n,\delta_n)$-ID code for channel $W^n$
is defined exactly in the same manner as in Section \ref{section:problem-formulation}.
We are interested in characterizing the doubly exponential optimal growth rate of the message size $N_n$.
\begin{definition}
For given $0 \le \ep, \delta < 1$, a rate $R$ is said to be $(\ep,\delta)$-achievable ID rate for general channel $\bm{W}$ if
there exists a sequence of $(N_n, \ep_n,\delta_n)$-ID codes satisfying
\begin{align}
\limsup_{n\to\infty} \ep_n &\le \ep, \label{eq:asymptotic-error-1} \\
\limsup_{n\to\infty} \delta_n &\le \delta, \label{eq:asymptotic-error-2}
\end{align}
and 
\begin{align}
\liminf_{n\to\infty} \frac{1}{n} \log \log N_n \ge R. \label{eq:asymptotic-rate}
\end{align}
Then, the supremum of $(\ep,\delta)$-achievable ID rates for $\bm{W}$ is 
termed the $(\ep,\delta)$-ID capacity, and is denoted by $C_{\mathtt{ID}}(\ep,\delta|\bm{W})$.
Particularly, for $(\ep,\delta)=(0,0)$, it is termed the ID capacity and denoted by $C_{\mathtt{ID}}(\bm{W})$.
\end{definition}

For a sequence $\bm{X} = \{ X^n \}_{n=1}^\infty$ of input processes, denote by $\bm{Y} = \{ Y^n\}_{n=1}^\infty$
the corresponding output processes via $\bm{W}=\{W^n\}_{n=1}^\infty$, i.e., $P_{Y^n} = P_{X^n} W^n$ for each $n$. 
Then, for $0 \le \varepsilon < 1$, let
\begin{align} 
\underline{I}^\varepsilon(\bm{X} \wedge \bm{Y}) := \sup\bigg\{ a : \limsup_{n\to\infty} \Pr\bigg( \frac{1}{n} \log \frac{W^n(Y^n|X^n)}{P_{Y^n}(Y^n)} \le a \bigg) \le \varepsilon \bigg\}.
   \label{eq:definition-epsilon-spectral-inf-mutual-information}
\end{align}
be the $\varepsilon$-spectral inf-mutual information rate. Particularly, when $\varepsilon = 0$, we just denote $\underline{I}(\bm{X} \wedge \bm{Y})$.

In \cite{hayashi:06b}, the following lower bound on the $(\ep,\delta)$-ID capacity was derived. 
\begin{proposition} \label{proposition:lower-bound-ID-capacity}
For $0 \le \ep, \delta < 1$ and a sequence $\bm{W}=\{ W^n \}_{n=1}^\infty$ of general channels, we have
\begin{align} \label{eq:lower-bound-general-capacity}
C_{\mathtt{ID}}(\ep,\delta | \bm{W}) \ge \sup_{\bm{X}} \underline{I}^\ep(\bm{X} \wedge \bm{Y}),
\end{align}
where the supremum is taken over all sequences of input processes $\bm{X}$.\footnote{Note that the right-side of \eqref{eq:lower-bound-general-capacity} does not
depend on $\delta$. Before \cite{hayashi:06b}, it had been known that $C_{\mathtt{ID}}(\ep,\ep | \bm{W})$ can be
lower bounded by the right-side of \eqref{eq:lower-bound-general-capacity} \cite{han:book}. } 
\end{proposition}

On the other hand, from Corollary \ref{corollary:single-shot-converse-beta-expression}, we can derive the following upper bound on the $(\ep,\delta)$-ID capacity.
\begin{theorem} \label{theorem:upper-bound-ID-capacity}
For $0 \le \ep, \delta < 1$ with $\ep+\delta<1$ and a sequence $\bm{W}=\{ W^n \}_{n=1}^\infty$ of general channels, we have
\begin{align*}
C_{\mathtt{ID}}(\ep,\delta | \bm{W}) \le \sup_{\bm{X}} \underline{I}^{\ep+\delta}(\bm{X} \wedge \bm{Y}).
\end{align*}
\end{theorem}
\begin{proof}
Suppose that $R$ is $(\ep,\delta)$-achievable ID rate, i.e., there exists a sequence of $(N_n,\ep_n,\delta_n)$-ID codes
satisfying \eqref{eq:asymptotic-error-1}, \eqref{eq:asymptotic-error-2}, and \eqref{eq:asymptotic-rate}.
By Corollary \ref{corollary:single-shot-converse-beta-expression}, we have
\begin{align}
\frac{1}{n} \log \log N_n \le \max_{P_{X^n}}\min_{Q_{Y^n}} -\frac{1}{n} \log \beta_{\ep_n+\delta_n+\eta_n}(P_{X^n}\times W^n, P_{X^n}\times Q_{Y^n}) + \Delta_n
\label{eq:proof-general-formula-1}
\end{align}
for $\eta_n=1/n$,\footnote{Since $\ep+\delta<1$, we have $\eta_n < 1-\ep_n-\delta_n$ for sufficiently large $n$.} where 
\begin{align*}
\Delta_n = \frac{1}{n}\big(\log n + \log\log|{\cal X}| + 2\log(1/\eta_n)+2 \big).
\end{align*}
Let $\hat{\bm{X}}=\{ \hat{X}^n\}$ be a sequence of input processes that attain the maximum in \eqref{eq:proof-general-formula-1} for each $n$,
and let $\hat{\bm{Y}}=\{ \hat{Y}^n\}$ be the corresponding output process. Then, we have
\begin{align}
\frac{1}{n} \log \log N_n \le -\frac{1}{n} \log \beta_{\ep_n+\delta_n+\eta_n}(P_{\hat{X}^n}\times W^n, P_{\hat{X}^n}\times P_{\hat{Y}^n}) + \Delta_n.
\label{eq:proof-general-formula-2}
\end{align}
Furthermore, by applying the righthand inequality of Lemma \ref{lemma:connection-Ds-beta}, we have
\begin{align}
\frac{1}{n} \log \log N_n \le -\frac{1}{n} \log D_{\mathtt{s}}^{\ep_n+\delta_n+2\eta_n}(P_{\hat{X}^n}\times W^n \| P_{\hat{X}^n}\times P_{\hat{Y}^n}) + \Delta_n + \frac{1}{n}\log(1/\eta_n)
\label{eq:proof-general-formula-2-2}
\end{align}
for sufficiently large $n$.

For arbitrary $\tau>0$, let $\xi= \underline{I}^{\ep+\delta}(\hat{\bm{X}} \wedge \hat{\bm{Y}}) + \tau$. Then, from the definition in \eqref{eq:definition-epsilon-spectral-inf-mutual-information},
there exists $\nu>0$ such that 
\begin{align*}
\Pr\bigg( \frac{1}{n} \log \frac{W^n(\hat{Y}^n|\hat{X}^n)}{P_{\hat{Y}^n}(\hat{Y}^n)} \le \xi \bigg) \ge \ep+\delta+ \nu
\end{align*}
for infinitely many $n$. Then, for those $n$'s, since $\limsup_{n\to\infty}\ep_n+\delta_n+2\eta_n \le \ep+\delta$, 
we have
\begin{align}
D_{\mathtt{s}}^{\ep_n+\delta_n+2\eta_n}(P_{\hat{X}^n}\times W^n \| P_{\hat{X}^n}\times P_{\hat{Y}^n})
\le \xi 
\label{eq:proof-general-formula-3}
\end{align}
provided that $n$ is sufficiently large. Thus, by combining \eqref{eq:proof-general-formula-2}, \eqref{eq:proof-general-formula-2-2} and \eqref{eq:proof-general-formula-3}, we have
\begin{align*}
R &\le \liminf_{n\to\infty} \frac{1}{n} \log\log N_n \\
&\le \xi \\
&\le \sup_{\bm{X}} \underline{I}^{\ep+\delta}(\bm{X} \wedge \bm{Y}) + \tau.
\end{align*} 
Since $\tau$ is arbitrary, any $(\ep,\delta)$-achievable ID rate $R$ must satisfy
$R \le \sup_{\bm{X}} \underline{I}^{\ep+\delta}(\bm{X} \wedge \bm{Y})$, which implies the claim of the theorem. 
\end{proof}

When the requirement of the type-II error probability is $\delta=0$, we can completely
characterize the ID capacity from Proposition \ref{proposition:lower-bound-ID-capacity} and Theorem \ref{theorem:upper-bound-ID-capacity} as follows. 
\begin{corollary}
For $0 \le \ep < 1$ and a sequence $\bm{W}=\{ W^n \}_{n=1}^\infty$ of general channels, we have
\begin{align*} 
C_{\mathtt{ID}}(\ep,0 | \bm{W}) = \sup_{\bm{X}} \underline{I}^\ep(\bm{X} \wedge \bm{Y}).
\end{align*}
Particularly, for $\ep=0$, we have
\begin{align} \label{eq:general-ID-capacity}
C_{\mathtt{ID}}(\bm{W}) = \sup_{\bm{X}} \underline{I}(\bm{X} \wedge \bm{Y}).
\end{align}
\end{corollary}

Note that \eqref{eq:general-ID-capacity} coincides with the general formula of the transmission capacity \cite{verdu:94}.
Thus, the ID capacity and the transmission capacity coincide for general channels. 
Previously, the coincidence of the ID capacity and the transmission capacity was known 
only for channels satisfying the strong converse property \cite{han:book}; it should be emphasized
that \eqref{eq:general-ID-capacity} holds without the assumption of the strong converse property. 

\section{Second-Order Coding Rate} \label{section:second-order}

In this section, we consider the second-order coding rate of the identification via discrete memoryless channels (DMCs) $W^n$.
As we have mentioned at the end of Section \ref{section:problem-formulation}, the optimal code size $N^\star(\ep,\delta | W^n)$
behaves like 
\begin{align*}
\log\log N^\star(\ep,\delta | W^n) = nC(W) + o(n),
\end{align*}
where $C(W)$ is the transmission capacity of channel $W$. 
In this section, we are interested in characterizing $L_{\mathtt{ID}}(\ep,\delta|W)$ in the expansion
\begin{align*}
\log\log N^\star(\ep,\delta | W^n) = nC(W) + \sqrt{n} L_{\mathtt{ID}}(\ep,\delta|W) + o(\sqrt{n})
\end{align*}
for fixed $0 < \ep < 1$ and vanishing $\delta \to 0$.

\begin{definition}
For given $0 < \ep,\delta < 1$ and DMC $W$, the second-order ID rate $L$ is defined to be $(\ep,\delta)$-achievable if there exists
a sequence of $(N_n,\ep_n,\delta_n)$-ID codes for $W^n$ satisfying 
\begin{align}
\limsup_{n\to\infty} \ep_n &\le \ep, \label{eq:definition-second-order-1} \\
\limsup_{n\to\infty} \delta_n &\le \delta, \label{eq:definition-second-order-2}
\end{align}
and 
\begin{align}
\liminf_{n\to\infty} \frac{1}{\sqrt{n}}\big( \log \log N_n - n C(W)  \big) \ge L. \label{eq:definition-second-order-3}
\end{align}
Then, the supremum of $(\ep,\delta)$-achievable second-order ID rates is termed the
second-order $(\ep,\delta)$-ID capacity, and is denoted by $L_{\mathtt{ID}}(\ep,\delta|W)$. Particularly, 
\begin{align*}
L_{\mathtt{ID}}(\ep|W) := \lim_{\delta\to0} L_{\mathtt{ID}}(\ep,\delta|W)
\end{align*}
is termed the second-order $\ep$-ID capacity. 
\end{definition}

In order to characterize the second-order rate, we need to introduce certain information quantities. 
Let 
\begin{align*}
\Pi(W) := \big\{ P_X \in {\cal P}({\cal X}) : I(X \wedge Y) = C(W) \big\}
\end{align*}
be the set of all capacity achieving input distributions. 
Even though capacity achieving input distributions may not be unique in general,
it is known that the capacity achieving output distribution $P_Y^*$ is unique. 

For a given output distribution $Q_Y$, let 
\begin{align*}
V(W \|Q_Y | P_X) := \sum_x P_X(x) \sum_y W(y|x) \bigg( \log \frac{W(y|x)}{Q_Y(y)} - D(W(\cdot |x) \| Q_Y) \bigg)^2
\end{align*}
be the conditional variance of the log-likelihood ratio between $W(\cdot|x)$ and $Q_Y$,
where $D(\cdot\|\cdot)$ is the KL-divergence.
Then, we define the minimum and the maximum of conditional information variances as
\begin{align*}
V_{\min}(W) &:= \min_{P_X \in \Pi(W)} V(W \| P_Y^* |P_X), \\
V_{\max}(W) &= \max_{P_X \in \Pi(W)} V(W \| P_Y^*|P_X).
\end{align*}
Using these quantities, $\ep$-dispersion of channel $W$ is defined as 
\begin{align*}
V_\ep(W) := \left\{
\begin{array}{ll}
V_{\min}(W) & \mbox{if } \ep < \frac{1}{2} \\
V_{\max}(W) & \mbox{if } \ep \ge \frac{1}{2}
\end{array}
\right. .
\end{align*}

For a given input distribution $P_X$ and corresponding output distribution $P_Y = P_XW$, let
\begin{align*}
U(P_X,W) := \sum_{x,y} P_X(x) W(y|x) \bigg( \log \frac{W(y|x)}{P_Y(y)} - I(X \wedge Y) \bigg)^2
\end{align*}
be the unconditional information variance. Then, we define the minimum and the maximum of 
unconditional information variances as
\begin{align*}
U_{\min}(W) &:= \min_{P_X \in \Pi(W)} U(P_X,W), \\
U_{\max}(W) &:= \max_{P_X \in \Pi(W)} U(P_X,W). \\
\end{align*}
Even though the unconditional information variance $U(P_X,W)$ can be strictly larger than 
the conditional information variance $V(W \| P_XW|P_X)$ in general, for capacity achieving input distributions,
these quantities coincide. Thus, the quantity 
\begin{align*}
U_\ep(W) := 
\left\{
\begin{array}{ll}
U_{\min}(W) & \mbox{if } \ep < \frac{1}{2} \\
U_{\max}(W) & \mbox{if } \ep \ge \frac{1}{2}
\end{array}
\right. .
\end{align*}
coincides with the $\ep$-dispersion $V_\ep(W)$ defined above \cite{Tan:book}. 

Now, we are ready to present the characterization of the second-order $\ep$-ID capacity. 
\begin{theorem} \label{theorem:second-order-rate}
For given DMC $W$ and $0 < \ep < 1$, if $V_\ep(W)>0$, then the second-order
$\ep$-ID capacity is given by
\begin{align} \label{eq:second-order-characterization}
L_{\mathtt{ID}}(\ep|W) = \sqrt{V_\ep(W)} \Phi^{-1}(\ep), 
\end{align}
where $\Phi^{-1}(\cdot)$ is the inverse function of the cumulative distribution
function 
\begin{align*}
\Phi(a)=\int_\infty^a \frac{1}{\sqrt{2\pi}} e^{-\frac{t^2}{2}} dt
\end{align*}
of the Gaussian distribution.
\end{theorem}

Note that the characterization of the second-order $\ep$-ID capacity in \eqref{eq:second-order-characterization}
coincides with the second-order $\ep$-transmission capacity \cite{Strassen:62, hayashi:09, polyanskiy:10}.

\subsection{Proof of achievability}

The achievability part of Theorem \ref{theorem:second-order-rate} is a straightforward consequence of
the achievability bound derived in \cite{hayashi:06b}.

\begin{lemma}[\cite{hayashi:06b}] \label{lemma:single-shot-achievability}
For given channel $W$ and input distribution $P_X$, let $P_Y$ be the corresponding output distribution. 
Assume that real numbers $a,a^\prime, b, b^\prime, \tau, \kappa>0$ satisfy
\begin{align} \label{eq:constrants-parameter-1}
\kappa \log \bigg( \frac{1}{\tau} - 1\bigg) > \log 2 + 1,~ 0 < \tau < 1/3,~0 < \kappa < 1
\end{align}
and
\begin{align} \label{eq:constrants-parameter-2}
1 > \frac{1}{a} + \frac{1}{a^\prime},~ c:= 1 - \frac{1}{b} - \frac{1}{b^\prime} > 0.
\end{align}
Then, for any integer $M>0$ and for any real number $K>0$, there exists an $(N,\ep,\delta)$-ID code such that
\begin{align*}
\ep &\le a b \Pr\bigg( \log \frac{W(Y|X)}{P_Y(Y)} \le \log K \bigg), \\
\delta &\le \kappa + a^\prime b^\prime \frac{1}{K} \bigg\lceil \frac{M}{c} \bigg\rceil, \\
N &= \bigg\lfloor \frac{e^{\tau M}}{M e} \bigg\rfloor
\end{align*}
provided that\footnote{In \cite[Eq.~(3)]{hayashi:06b}, there is a typo that $\alpha$ ($a$ in our notation) is missing.}
\begin{align*} 
ab \Pr\bigg( \log \frac{W(Y|X)}{P_Y(Y)} \le \log K \bigg) + a^\prime b^\prime \frac{1}{K} \bigg\lceil \frac{M}{c} \bigg\rceil < 1,
\end{align*}
where $(X,Y) \sim P_X \times W$.
\end{lemma}

Now, we go back to the proof of achievability. 
For a given $0<\ep<1$, fix a capacity achieving input distribution $P_X$ that attains $U_\ep(W)$;
then, let $P_Y$ be the corresponding output distribution of channel $W$.
By setting $a=b=1+\frac{2}{n}$, $a^\prime=b^\prime=(n+2)$, $\tau=\frac{1}{n+2}$, and $\kappa=\frac{1+\log 2}{\log n}$,
we can verify that the conditions in \eqref{eq:constrants-parameter-1} and \eqref{eq:constrants-parameter-2} are satisfied for $n\ge 2$. 
For $R>0$, we apply Lemma \ref{lemma:single-shot-achievability} by setting $K=e^{nR}$ and $M = \lceil e^{nR}/(n+2)^4 \rceil$;
then, there exist a constant $F>0$ and a sequence of $(N_n,\ep_n,\delta_n)$-ID codes such that
\begin{align*}
\frac{1}{n} \log \log N_n \ge R - \frac{F}{n}\log n
\end{align*}
and
\begin{align*}
\ep_n &\le \bigg( 1 + \frac{2}{n} \bigg)^2 \Pr\bigg( \frac{1}{n} \log \frac{W^n(Y^n|X^n)}{P_Y^n(Y^n)} \le R \bigg), \\
\delta_n &\le \frac{1+\log 2}{\log n} + \frac{2}{n+2}
\end{align*}
provided that
\begin{align} \label{eq:constrants-parameter-4}
\bigg( 1 + \frac{2}{n} \bigg)^2 \Pr\bigg( \frac{1}{n} \log \frac{W^n(Y^n|X^n)}{P_Y^n(Y^n)} \le R \bigg) + \frac{2}{n+2} < 1,
\end{align}
where $(X^n, Y^n) \sim P_X^n \times W^n$.
Here, set 
\begin{align*}
R = C(W) + \sqrt{\frac{U_\ep(W)}{n}} \Phi^{-1}(\ep).
\end{align*}
Then, by applying the central limit theorem, we have
\begin{align*}
\lim_{n\to\infty} \Pr\bigg( \frac{1}{n} \log \frac{W^n(Y^n|X^n)}{P_Y^n(Y^n)} \le R \bigg) = \ep.
\end{align*}
Thus, the condition in \eqref{eq:constrants-parameter-4} is satisfied for sufficiently large $n$,
and there exists a sequence of $(N_n,\ep_n,\delta_n)$-ID codes satisfying \eqref{eq:definition-second-order-1}-\eqref{eq:definition-second-order-3} 
for $L=\sqrt{U_\ep(W)}\Phi^{-1}(\ep)$ and an arbitrary $\delta>0$. Thus, we have
\begin{align*}
L_{\mathtt{ID}}(\ep|W) &\ge \sqrt{U_\ep(W)}\Phi^{-1}(\ep) \\
&= \sqrt{V_\ep(W)}\Phi^{-1}(\ep),
\end{align*}
which completes the proof of the achievability part of Theorem \ref{theorem:second-order-rate}. \qed

\subsection{Proof of converse}

By Corollary \ref{corollary:single-shot-converse-Ds-expression} and the symbol-wise relaxation \eqref{eq:symbol-wise-relaxation}, we have
\begin{align}
\log\log N^\star(\ep,\delta|W^n) \le \inf_{Q_n \in {\cal P}({\cal Y}^n)} \max_{x^n \in \cX^n}
D_{\mathtt{s}}^{\ep+\delta+\eta}(W^n(\cdot|x^n ) \| Q_n)
+ \log \log |{\cal X}^n| + 2 \log(1/\eta) + 2. 
\label{eq:second-order-converse-proof-1}
\end{align}
Since the terms other than the first one in \eqref{eq:second-order-converse-proof-1} are
$o(\sqrt{n})$, the remaining task is to evaluate the first term of \eqref{eq:second-order-converse-proof-1}
for an appropriate choice of the output distribution $Q_n$. 
For the purpose of deriving the second-order rate, it suffices to choose a mixture 
of the capacity achieving output distribution and output distributions induced from types on $\cX^n$ \cite{hayashi:09}. 
Although it is more than necessary to derive the second-order rate, we refer to a stronger result that is derived 
by a more sophisticated choice of the output distribution \cite{TomTan:13}.
\begin{lemma}[\cite{TomTan:13}] \label{lemma:converse-evaluation-D}
Suppose that $V_{\ep+\delta}>0$. For $\eta=1/\sqrt{n}$, there exists a constant $F$ such that 
\begin{align*}
\inf_{Q_n \in {\cal P}({\cal Y}^n)} \max_{x^n \in \cX^n} D_{\mathtt{s}}^{\ep+\delta+\eta}(W^n(\cdot|x^n ) \| Q_n) \le n C(W) + \sqrt{n V_{\ep+\delta}} \Phi^{-1}(\ep+\delta) + F
\end{align*} 
for sufficiently large $n$. 
\end{lemma}

By \eqref{eq:second-order-converse-proof-1} and Lemma \ref{lemma:converse-evaluation-D}, we have
\begin{align*}
L_{\mathtt{ID}}(\ep,\delta|W) \le \sqrt{V_{\ep+\delta}} \Phi^{-1}(\ep+\delta).
\end{align*}
Finally, by taking the limit of $\delta \to 0$, we have the converse part of Theorem \ref{theorem:second-order-rate}. \qed

\section{Discussion}

In this paper, we have derived a minimax converse bound for the identification via channel.
By using this converse bound, we have derived the general formula for the identification capacity without the assumption
of the strong converse property; the problem has been unsolved for a long time.  
Our converse is built upon the partial channel resolvability introduced in \cite{steinberg:98}.
When we derive the converse bound for the identification code using the channel resolvability, a crucial observation
is the counting argument in which the number of messages of the identification code is bounded by the number of $M$-types. 
Even though the partial channel resolvability approach have had a potential to improve the bound based on 
the channel resolvability, there was a difficulty that the counting argument does not work for the partial channel resolvability, at least without an additional trick;
cf.~Remark \ref{remark:proof-argument}.
We have overcome this difficulty by utilizing the auxiliary output distribution, 
the idea that has become popular in the past decade. 
As a future direction, it is tempting to apply the auxiliary output distribution approach
to other problems of identification code.

\section*{Acknowledgment}

The author would like to thank Yasutada Oohama for a fruitful discussion. This work was supported in part by
JSPS KAKENHI under Grant 20H02144.

\bibliographystyle{./IEEEtranS}
\bibliography{../../../09-04-17-bibtex/reference}

\end{document}